\documentclass{article}

     \usepackage[preprint]{neurips_2024}

\usepackage{natbib}

\usepackage[utf8]{inputenc} 
\usepackage[T1]{fontenc}    
\usepackage{hyperref}       
\usepackage{url}            
\usepackage{booktabs}       
\usepackage{amsfonts, amsmath, amssymb, bbm, amsthm}       
\usepackage{nicefrac}       
\usepackage{microtype}      
\usepackage{xcolor}         
\definecolor{auburnblue}{RGB}{12, 35, 64}
\definecolor{auburnorange}{RGB}{232, 119, 34}

\usepackage{caption}
\usepackage{subcaption}
\usepackage{float} 
\usepackage{pgfplots}
\pgfplotsset{compat=1.17}
\usepgfplotslibrary{groupplots}
\usepackage{pgfplotstable}
\usepackage{filecontents}
\usepackage{tikz}
\usepackage{algorithm}
\usepackage{algpseudocode}
\usepackage{tabularx} 
\usepackage{bm}
\usepackage{todonotes}  

\newtheorem{theorem}{Theorem}
\newtheorem{prop}{Proposition}

\newtheorem{remk}{Remark}

\newtheorem{assm}{Assumption}
\newcommand{\argzero}{\mbox{argzero}}
\newcommand{\argmin}{\mbox{argmin}}

\title{Fiducial Matching:\\ Differentially Private Inference for Categorical Data}

\author{%
  Ogonnaya Michael Romanus\\ 
  Auburn University\\
  Auburn, AL 36849 \\
  \texttt{omr0010@auburn.edu} \\
   \And
  Younes Boulaguiem\\
  University of Geneva \\
  Geneva 1205, Switzerland \\
  \texttt{younes.boulaguiem@unige.ch} \\
  \AND
  Roberto Molinari \\
  Auburn University \\
  Auburn, AL 36849 \\
  \texttt{robmolinari@auburn.edu} \\
}

\begin{document}

\maketitle

\begin{abstract}
    The task of statistical inference, which includes the building of confidence intervals and tests for parameters and effects of interest to a researcher, is still an open area of investigation in a differentially private (DP) setting. Indeed, in addition to the randomness due to data sampling, DP delivers another source of randomness consisting of the noise added to protect an individual's data from being disclosed to a potential attacker. As a result of this convolution of noises, in many cases it is too complicated to determine the stochastic behavior of the statistics and parameters resulting from a DP procedure. In this work, we contribute to this line of investigation by employing a simulation-based matching approach, solved through tools from the fiducial framework, which aims to replicate the data generation pipeline (including the DP step) and retrieve an approximate distribution of the estimates resulting from this pipeline. For this purpose, we focus on the analysis of categorical (nominal) data that is common in national surveys, for which sensitivity is naturally defined, and on additive privacy mechanisms. We prove the validity of the proposed approach in terms of coverage and highlight its good computational and statistical performance for different inferential tasks in simulated and applied data settings.
\end{abstract}

\section{Introduction}

Differential Privacy (DP), along with its various adaptations in recent years, aims to safeguard individual privacy by providing a mathematically rigorous guarantee of protection \cite{c20}. This guarantee is achieved in practice by injecting probabilistic (random) noise into different stages of the data processing and analysis pipeline. The privacy mechanisms implementing DP can act directly on the raw data itself or on the outputs derived from the data, such as summary statistics, model parameters, or query results. By quantifying the privacy risk in the presence of arbitrary external information, DP ensures strong protection against re-identification attacks. However, this privacy-preserving approach comes at a cost: the introduction of noise inevitably affects the utility of the data output. The extent of this impact varies depending on the analysis task and the specific privacy mechanism used, but generally, DP induces statistical bias, increases variability, and complicates statistical inference by affecting key tasks such as hypothesis testing and uncertainty quantification \cite{c6, c22}. 

A fundamental characteristic of DP is that it implies that the data is affected by two distinct sources of randomness: the inherent randomness from sampling of the original dataset and the additional randomness imposed by the privacy mechanism. This dual source of randomness poses significant challenges for statistical inference. Once noise is introduced to satisfy privacy constraints, the original data become inaccessible to the analyst, meaning that standard inferential procedures based on classical likelihood theory may no longer be applicable. Moreover, the sampling distribution of privatized data can become highly complex and difficult to characterize analytically, particularly when noise is introduced in a non-trivial manner \citep{williams2010probabilistic}. As a result, standard statistical methods may not account for the distortions introduced by DP, leading to biased estimators, inaccurate confidence intervals, and unreliable hypothesis testing results.

To illustrate these challenges, consider a setting in which we have a probabilistic model \( f_{\bm{\theta}}\) that describes the data-generating process for the sensitive dataset \( \mathcal{D} \), where \( \bm{\theta} \in \Theta \subset \mathbb{R}^p \) represents the parameter of interest. Furthermore, suppose that a privacy mechanism \( m(s|\mathcal{D} ) \) is applied to privatize the data, producing the observed privatized summary \( s \). A natural approach to make an inference on \( \bm{\theta} \) in this setting would be to compute the marginal likelihood of \( \bm{\theta} \) given \( s \). However, computing this marginal likelihood is often infeasible in practice, as it requires integrating over all possible realizations of the original dataset \( \mathcal{D}  \) \cite{c21}. This intractability increases when dealing with high-dimensional data and with the complexity of the privacy mechanism, particularly when the noise is non-Gaussian or correlated. Consequently, conventional likelihood-based inference methods may be impractical within the DP framework. 

Given these limitations, inference methods that bypass the computation of the likelihood have been increasingly investigated for DP-constrained statistical analysis. Among these, simulation-based inference techniques, such as parametric bootstrap and approximate Bayesian computation, have been widely explored to handle uncertainty quantification in DP settings \cite{c26,c23,c24,c25}. These methods leverage the ability to sample from the generative model \( f_{\bm{\theta}}(\mathcal{D} ) \) and the privacy mechanism \( m(s|\mathcal{D} ) \) to approximate posterior distributions or construct valid confidence intervals without requiring explicit likelihood evaluation. Moreover, simulation-based inference methods align naturally with the DP framework, since the privacy mechanism \( m(s|\mathcal{D} ) \) can be made publicly available without compromising individual privacy. This transparency allows for the development of robust inference procedures while ensuring that privacy constraints are rigorously maintained. 

However, these methods also introduce computational burdens as they typically require extensive sampling to achieve accurate approximations. With this in mind, also in this work we approach the problem of statistical inference on privatized categorical data through a different paradigm than the likelihood-based one, and notably we use the simulated moment-matching (extremum estimator) framework \cite[see e.g.][]{c13}. In fact, we make use of the approach proposed in \cite{orso2024accurate} which is based on the statistical properties of moment-matching simulation-based approaches that were studied and refined in \cite{c14}, \cite{c15} and \cite{c16}. In this work, we therefore adopt this moment-matching approach to deliver valid DP confidence intervals (and hypothesis tests) for parameters of interest. In particular, given the discrete nature of the matching solution, we use a fiducial argument to solve the optimization problem \citep{c25}. In contrast to existing methods which are specific to certain types of statistical tasks, this proposed approach can be applied to a wider variety of applications on categorical data, such as, among others, two sample tests for proportions, chi-squared tests or (saturated) logistic models with categorical predictors \citep[see e.g.][]{agresti2013categorical}.

\section{Related Work}

Over the past decade, significant research has focused on statistical inference within the framework of DP. For example, the Bayesian framework offers a highly flexible approach for analyzing privatized data. \citet{williams2010probabilistic} were the first to highlight that, as mentioned earlier, the marginal likelihood is computationally intractable. To address this challenge, various data augmentation MCMC methods have been introduced. For example, \citet{bernstein2018differentially,bernstein2019differentially} demonstrated that when the distribution of a sufficient statistic is approximated as a normal distribution, a simple Gibbs sampler can be used to estimate the posterior distribution. Meanwhile, \citet{ju2022data} developed a Metropolis-within-Gibbs sampler that directly targets the exact private posterior distribution, although its effectiveness is limited unless conjugate priors and small sample sizes are used. For simpler cases, \citet{gong2019exact} introduced an independent and identically distributed (i.i.d.) sampling method based on rejection sampling. Alternatively, \citet{karwa2015private} proposed a variational approximation specifically tailored for a naive Bayes classifier.

Beyond Bayesian methods, frequentist approaches have also been explored for statistical inference on privatized data. \citet{wang2018statistical} argued that conventional asymptotic approximations often perform poorly in finite samples when applied to privatized data. To address this, they proposed a modified asymptotic framework in which only the sufficient statistics are subject to asymptotic approximations, while the noise introduced by the privacy mechanism remains unaltered. Another widely used frequentist technique involves the parametric bootstrap, which approximates the sampling distribution based on privatized data \citep{ferrando2022parametric,alabi2022hypothesis}. However, \citet{awan2024simulation} demonstrated that the parametric bootstrap can lead to biased inference, resulting in inadequate coverage and incorrect type I error rates. As a more reliable alternative, they introduced a simulation-based inference approach that ensures conservative control over coverage and type I errors, although at the cost of increased computational complexity.

Specifically, discussing categorical (nominal) data, \citet{c6} were among the first to discuss statistical hypothesis testing within a DP framework in these settings. In particular, using the Laplace DP mechanism, they studied the performance of standard statistical tests (e.g., tests for proportions and independence in $2\times2$ tables) under the constraint of DP showing how these are highly affected by DP mechanisms and that large sample sizes are needed to achieve reasonable performance. Additionally, they utilized normal approximation to the sampling distribution to calibrate the type I error. Building on this work, \cite{c5} expanded these techniques to multinomial data, employing Monte Carlo methods to control and estimate the type I error. In contrast, \cite{C7} formulated asymptotic distributions for their tests and validated the reliability of the type I error through simulations. \cite{c3} devised private chi-square tests for goodness-of-fit and identity problems, ensuring that the asymptotic distributions of the test statistics, after accounting for added noise for privacy, align with those of classical chi-square tests. Furthermore, \cite{c8} developed DP p-values for chi-squared tests of GWAS data and derived the exact sampling distribution of the noisy statistics. 

To our knowledge, the first formal inference framework for categorical data was the uniformly most powerful (UMP) hypothesis test under DP for binomial data proposed in \cite{c9}. Their testing procedure relies on the Tulap distribution, which is a combination of discrete Laplace and uniform distributions. They also present an algorithm for generating a $p$-value to access one-sided hypotheses, demonstrating that the resulting $p$-value is the smallest $\epsilon$-DP $p$-value for this test. Simultaneously, \cite{c11} and \cite{c12} introduced a similar framework for constructing a differentially private version of any non-private hypothesis test in a black-box manner, employing the Subsample-and-Aggregate technique proposed by \cite{c10}. Despite both methods starting with subsampling, the aggregation strategies differ. \cite{c11} used the algorithm of \cite{c9} as the aggregation function, while \cite{c12} designed a custom binomial test based on a randomized response-type method.

\section{Preliminaries}

Let $X \sim f_{\bm{\theta}}^K$, where $f_{\bm{\theta}}^K$ is a multinomial distribution with $K$ categories (classes) and, as mentioned earlier, $\bm{\theta}$ is the parameter of interest that characterizes the distribution $f$. Typically $\bm{\theta}~\in~\Delta^{K-1}~\subset~\mathbb{R}^{K}$ is a vector defining the probabilities for each of the $K$ classes and belonging to the $(K-1)$-dimensional simplex:
\[
\Theta := \left\{ \bm{\theta} \in \mathbb{R}_+^{K} : \sum_{k=1}^{K} \theta_k = 1, \theta_k \geq 0 \,\, \text{for} \,\, k = 1, \hdots, K \right\}.
\]
More specifically, we have that $\theta_k = \mathbb{P}(X = k)$ for $k = 1, \hdots, K$, which can also be expressed through other parameterizations based on known deterministic transformations of $\bm{\theta}$  (for example, log-odds in logistic models). Based on this, letting $\mathcal{D} := \{x_1, \hdots, x_n\}$ represent the data with $n$ samples, the maximum-likelihood (best unbiased) estimator for $\theta_k$ is given by the sample proportion:
\begin{equation}
\label{eq:sample_prop}
    \bar{\theta}_k = \frac{1}{n} \sum_{i = 1}^n \mathbbm{1}(X_i = k),
\end{equation}
which therefore represents the estimated probability of $X$ falling into class $k$. Without loss of generality, for the following sections, let us consider the case of $K = 2$ such that $X \in \{0, 1\}$ which implies that $f_{\theta}$ is a Binomial distribution with parameter $\theta$ corresponding to the probability of success (i.e. $\theta = \mathbb{P}(X = 1)$) with estimator $\bar{\theta}$. We will then extend our study to the case $K > 2$ further on.

Once we have the estimator $\bar{\theta}$, a common task is to build Confidence Intervals (CIs) for the true parameter, say $\theta_0$, or to run some hypothesis tests on the true parameter using this estimator. To do so, the sampling distribution of the estimator is required, and, in the case of the sample proportion, there are exact as well as asymptotic distributions that are commonly employed to achieve appropriate coverage and type I errors. However, under the constraint of pure $\epsilon$-DP we cannot release the observed sample proportion, and a common mechanism used to guarantee  $\epsilon$-DP in these settings is to release the statistic delivered from the following additive mechanism:
\begin{equation}
\label{eq.dp_prop}
    \hat{\pi} := \bar{\theta} + \frac{\Delta_{\bar{\theta}}}{\epsilon} Z,
\end{equation}
where $Z$ is a zero-mean and unit variance random variable independent of the data; $\epsilon$ is the privacy budget and $\Delta_{\bar{\theta}}=\sup_{\mathcal{D},\mathcal{D}'} \lVert \bar{\theta}(\mathcal{D})-\bar{\theta}(\mathcal{D}')\rVert$ is the sensitivity of the sample proportion $\bar{\theta}$ with respect to a norm $\lVert\cdot\rVert$ computed on $n$ samples \citep[see e.g.][]{dwork2006differential, dwork2006calibrating}. More in detail, $\mathcal{D}$ and $\mathcal{D}'$ are two neighboring databases of size $n$ (i.e. databases that differ only by one entry) and $\bar{\theta}(\mathcal{D})$ represents the sample proportion computed on database $\mathcal{D}$. Hence, the sensitivity $\Delta_{\bar{\theta}}$ measures the maximum amount by which the sample proportion can change by arbitrarily changing one observation from the database, which in the case of the sample proportion is notably $\Delta_{\bar{\theta}} = \nicefrac{1}{n}$. On this note, in this work we only consider \textit{additive} DP mechanisms of the form in \eqref{eq.dp_prop} and assume that there are no privacy constraints on the data size $n$, as is the case for the US Census and other surveys, for example. To satisfy $\epsilon$-DP, the $L_1$ norm is used and $Z$ is a vector with entries coming from $\text{Laplace}(0,1)$. To satisfy $\epsilon$-Gaussian Differential Privacy (GDP) the $L_2$ norm is used and $Z\sim N(0,1)$.  As a side note, it was shown in \cite{geng2015optimal} that adding discrete Laplace noise is optimal under $\epsilon$-DP; however, in this work we want to allow for more general definitions of DP. Hereinafter, for notational simplicity, we denote the privacy noise as $Y:= \nicefrac{\Delta_{\bar{\theta}} Z}{\epsilon}$ (for any appropriate choice of $Z$). Based on this, we have our first assumption for our approach proposed later in Section \ref{sec.fima}.\\

\begin{assm}
\label{asm.additive}
    The privacy mechanism to ensure DP is an additive noise mechanism as in \eqref{eq.dp_prop}.
\end{assm}

In this work we will use notation that aligns with the definition of $\epsilon$-DP. However, as long as the mechanism remains an additive one, any form of DP can be addressed by the proposed methodology described further on. This being said, with the addition of the privacy noise it is straightforward to see that $\hat{\pi} = \bar{\theta} + Y$ is an unbiased and consistent estimator for the true parameter of interest $\theta_0$. However, it may not be straightforward to derive or approximate the non-asymptotic sampling distribution of $\hat{\pi}$ given the convolution of two sources of randomness (sampling and privacy noise). If using discrete Laplace noise for the one-sample proportion setting, then \cite{c9} derived the ``Tulap'' distribution which allowed them to deliver the UMP test for binomial data (along with CIs and similar inferential tools). Under other privacy mechanisms, alternative solutions can be conceived, but they have different limitations, as highlighted earlier. For example, the parametric bootstrap \cite{efron2012bayesian} could generally be a potential solution but it would not work if $\hat{\pi} \notin [0, 1]$, which is indeed possible for different choices of $Y$, and it may not deliver the correct distribution since this depends on the point in which $\hat{\pi}$ lands on the interval $[0, 1]$.

To overcome the above issues, we consider an alternative simulation-based framework proposed in \cite{orso2024accurate} which is inspired by indirect inference \cite{c13} and akin to the repro sampling technique in \cite{xie2022repro}. More specifically, we define a new DP estimator based on a matching criterion and use tools from the fiducial inference literature to deliver solutions for it \citep{c27}. Under certain regularity conditions, the bootstrapped solutions for this estimator are valid when employing the percentile approach.

\section{Methodology}

To describe the proposed approach taken from \cite{orso2024accurate}, we first abuse notation slightly by defining the categorical variable as $X := X(\theta_0)$ to underline its implicit dependence on the true parameter $\theta_0$ through the distribution $f_{\theta_0}$. This being said, our goal is to construct valid CIs for $\theta_0 \in [0,1]$ knowing that the initial estimator $\hat{\pi}$, computed on $X(\theta_0)$, converges in probability to the non-stochastic limit $\theta_0$. To build such intervals, let us express the data-generating mechanism as $X(\theta) = g(\theta, W)$, where $W \in \mathbb{R}^m$ is a random vector whose distribution does not depend on $\theta$. For example, $W$ can be seen as a vector of iid standard uniform random variables such as those used to simulate samples from parametric distributions through their inverse cumulative distribution functions. Following this expression, we can also express \textit{simulated} data as $X^*(\theta) = g(\theta, W^*)$ where $W^*$ represents an independent copy of $W$. As a consequence, $\hat{\pi}$ represents the DP estimator computed on the observed data $X(\theta_0, w)$, where $w$ is fixed and unknown for the observed data, while we let $\hat{\pi}^*(\theta)$ represent the DP estimator computed on the simulated data $X^*(\theta, w^*)$, where $w^*$ is a sample of the random variable $W^*$. 

Based on the setup of the above problem, a first possible framework that one may consider to find a distribution with respect to the parameter $\theta$ would be that of (generalized) Fiducial inference \cite{c27}. The latter aims to find a distribution by matching the observed data with simulated data or, under certain conditions, by conditioning on sufficient statistics for $\theta$. However, ignoring ongoing conceptual debates regarding this framework, in our DP setting we do not have access to the data, and it is not clear if the DP estimator $\hat{\pi}$ qualifies as a sufficient statistic once privacy noise has been added (especially if $\hat{\pi} \notin (0,1)$). We therefore consider an alternative framework which defines an estimator for the quantity of interest starting from auxiliary estimators (including biased and/or inconsistent estimators for $\theta_0$) and employs tools from the fiducial literature to deliver solutions for this estimator and valid (frequentist) statistical coverage of $\theta_0$. Indeed, this framework employs general statistics related to the parameter $\theta_0$ and the proposed estimator for this parameter is given by the solution to the optimization problem:
\begin{equation}
\label{eq:ib}
    \hat{\theta} \in \underset{\theta \in [0,1]}{\argmin} \, \| \hat{\pi} - \hat{\pi}^*(\theta) \|,
\end{equation}
where $\| \cdot \|$ denotes the Euclidean norm. With this definition, recalling our previous notation, it is clear that if $w^* = w$, then $\theta_0$ is guaranteed to be a solution of \eqref{eq:ib}. In practice, the distribution of $\hat{\theta}$ is approximated by generating $H$ independent realizations of the random variable $W^*$, with $H$ generally being large, which simply amounts to simulating $H$ datasets of size $n$ from the distribution $f_{\theta}$. For each of these, the DP estimator $\hat{\pi}^*(\theta)$ is computed by also generating new realizations of the privacy noise $Y^*$ (whose parameters are known), while $\hat{\pi}$ is fixed from the observed data for all $h = 1, \hdots, H$. Based on each of these replicates, the matching problem in \eqref{eq:ib} is solved $H$ times thereby generating $H$ samples of the random variable $\hat{\theta}$, that is, $\{\hat{\theta}_h\}_{h = 1, \hdots, H}$, on which the percentile method can be used to obtain one- or two-sided CIs \citep[more details in][]{orso2024accurate}.

However, two major issues remain in obtaining this distribution: (i) the matching problem in \eqref{eq:ib} may not be computationally efficient to solve and (ii) it is not clear how to sample $\hat{\theta}$ given the discrete nature of the data $\mathcal{D}$. For this purpose, by reformulating the matching problem in \eqref{eq:ib}, we present a solution which relies on tools from the fiducial framework.

\subsection{FIMA: Fiducial Matching}
\label{sec.fima}

Let us start by rewriting the matching problem in \eqref{eq:ib} as an ``$\argzero$'' problem, i.e.
$$\hat{\theta} \in \underset{\theta \in [0,1]}{\argzero} \,[\hat{\pi} - \hat{\pi}^*(\theta)].$$
Essentially, this implies that we want to find the value of $\theta$ such that
\begin{equation}
    \label{eq.match_arg}
    \hat{\pi} = \hat{\pi}^*(\theta).
\end{equation}
We can expand the second term of this matching as follows: 
$$\hat{\pi}^*(\theta) := \bar{\theta}^*(\theta) + Y^*,$$
where $\bar{\theta}^*(\theta)$ is the simulated sample proportion implied by $\theta$ and $Y^*$ is an independent copy of the privacy noise $Y$ (by definition independent of the distribution $f_{\theta}$). This allows us to rewrite \eqref{eq.match_arg} as follows:
\begin{equation}
\label{eq:alt_ib}
    \hat{\pi} = \bar{\theta}^*(\theta) + Y^* \implies \bar{\theta}^*(\theta) = \hat{\pi} - Y^*.
\end{equation}
Let us now focus on the second equation in \eqref{eq:alt_ib} where it can be noted that the simulated sample proportion $\bar{\theta}^*(\theta)$ can also be expressed as
$$\bar{\theta}^*(\theta) := \frac{1}{n}\sum_{i=1}^n \mathbbm{1}(U_i^* \leq \theta),$$
with $U_i^* \overset{iid}{\sim} U(0,1)$ being an iid standard uniform random variable, for $i = 1, \hdots, n$ \citep{c27}. Hence, $\bar{\theta}^*(\theta)$ corresponds to an empirical cumulative distribution function (CDF) based on the sample $U_1^*, \dots, U_n^*$, denoted as $F_{U^*}(\theta)$. As a result, the second equality in \eqref{eq:alt_ib} can be re-expressed as
$$F_{U^*}(\theta) = \hat{\pi} - Y^*,$$
such that, defining $\tilde{\theta}^* := \hat{\pi} - Y^*$, we have that the solution to \eqref{eq.match_arg} is given by
\begin{equation}
\label{eq.ib_sol}
    \hat{\theta} = F_{U^*}^{-1}(\tilde{\theta}^*).
\end{equation}
Therefore, to obtain this solution we need to inverse the empirical distribution function $F_U^*$ which is discrete by nature, implying that we do not obtain point solutions but intervals. Indeed, the interval solutions are given by:
\begin{equation*}
\label{eq:int_sol}
    I(\hat{\pi}, Y^*, \{U_i^*\}) =
\begin{cases}
    \delta & \text{if }\, \tilde{\theta}^*\leq0\\
    [\delta, U_{[1]}^*) & \text{if }\, 0 \leq \tilde{\theta}^*< \frac{1}{n}\\
    [U_{[\lfloor n\tilde{\theta}^* \rfloor]}^*, U_{[\lfloor n\tilde{\theta}^* +1 \rfloor]}^*) & \text{if }\, \frac{1}{n} \leq \tilde{\theta}^*<1-\frac{1}{n}\\
    [U_{[n]}^*, 1 - \delta) & \text{if }\, 1-\frac{1}{n} \leq \tilde{\theta}^* < 1\\
    1 - \delta & \text{if } \, \tilde{\theta}^*\geq 1
\end{cases}\,\,\,\,\,,
\end{equation*}
where $U_{[i]}^*$ represents the $i^{\text{th}}$ ordered value of the sequence of standard uniform variables $\{U_i^*\}~\in~[0,1]^n$, $\lfloor v\rfloor$ represents the largest integer smaller than or equal to $v$ and $0 < \delta \ll 1$ is a small constant close to zero (for example, the machine epsilon). As a result, sampling a value from the interval solution $I(\hat{\pi}, Y^*, \{U_i^*\})$ consists in a solution to the matching problem in \eqref{eq:ib}. Indeed, following \cite{c27}, we know that a \textit{generalized fiducial quantity} for this problem is given by:
\begin{equation*}
    \hat{\theta} = U_{[\lfloor n\tilde{\theta}^* \rfloor]}^* + D\,\left(U_{[\lfloor n\tilde{\theta}^* \rfloor]}^* - U_{[\lfloor n\tilde{\theta}^* +1 \rfloor]}^*\right),
\end{equation*}
where $D$ is a random variable on the support $[0, 1]$. Based on \cite{c27}, any choice of distribution for $D$ (under the stated support constraint) provides a fiducial quantity consisting in a solution to \eqref{eq.ib_sol}. The final procedure, justified on the basis of the above steps, is presented in Algorithm \ref{algo_ifb} and we refer to it as Fiducial Matching (FIMA).

\vspace{0.2cm} \noindent 
\begin{algorithm}[H]
\caption{Fiducial Matching (FIMA)}
\label{algo_ifb}
\begin{algorithmic}[1]
    \State \textbf{Input:} $\hat{\pi}$: DP proportion/probability; $n$: data size; $\epsilon$: privacy budget; $H$: number of solutions; $q$: distribution for $Y$; $g$: distribution for $D$; $\delta > 0$: small quantity close to zero.
    \For{$h = 1, \hdots, H$}
    \State $\tilde{\theta}_h^* = \hat{\pi} - Y_h^*$, where $Y_h^* \sim q\left(\frac{1}{n \epsilon}\right)$.
    \If{$\tilde{\theta}_h^* > 1$}
        \State \textbf{Return} $\hat{\theta}_h = 1 - \delta$
    \ElsIf{$\tilde{\theta}_h^* < 0$}
        \State \textbf{Return} $\hat{\theta}_h = \delta$
    \Else
        \State Generate $\{U_{h,i}^*\}$, for $i=, 1, \hdots, n$ and $U_{h,i}^* \overset{iid}{\sim} U(0,1)$
        \State \textbf{Return} $\hat{\theta}_h = U_{[\lfloor n\tilde{\theta}_h^* \rfloor]}^* + D_h\,(U_{[\lfloor n\tilde{\theta}_h^* \rfloor]}^* - U_{[\lfloor n\tilde{\theta}_h^* +1 \rfloor]}^*)$, where $D_h \sim g$
    \EndIf
    \EndFor
    \State \textbf{Output:} A sequence $\{\hat{\theta}_h\}$ for $h = 1, \hdots, H$
\end{algorithmic}
\end{algorithm}

The FIMA described in Algorithm \ref{algo_ifb} is a logical representation of the procedure but can of course be made computationally more efficient by avoiding the loop and generating all random values at the first step (i.e. $Y^*$, $\{U_i^*\}$ and $D$ for all $n$ and $H$). Thereafter, all that is needed is to perform indexing on predefined objects since $\hat{\pi}$ is fixed, and consequently $\{\tilde{\theta}_h^*\}$ can be computed in one vectorized operation. Moreover, in this work we will make the choice of $D \sim \text{Beta}(\nicefrac{1}{2}, \nicefrac{1}{2})$ which directly delivers the FIMA solution:
$$\hat{\theta} \sim \text{Beta}(n\tilde{\theta}^* +\nicefrac{1}{2}, n-n\tilde{\theta}^*+\nicefrac{1}{2}),$$
consisting in a posterior distribution based on Jeffrey's prior \citep[see again][]{c27}. Hence, using this solution, there is no need to generate the realizations $D$ or the sequences $\{U_h^*\}$, thus also avoiding the ordering operation on the latter.

\begin{remk}
    An alternative solution to sample a point estimate $\hat{\theta}$ could be to approximate the distribution of $\hat{\pi}^*(\theta) = \bar{\theta}^*(\theta) + Y$ (and hence also $\hat{\pi}$) using a central limit theorem for $\bar{\theta}^*(\theta)$ thereby obtaining
    $$\hat{\pi}^* \overset{d}{\approx} \theta + \sqrt{\frac{\theta(1 - \theta)}{n}}V^* + Y^*,$$
    where $V^* \sim \mathcal{N}(0,1)$. This would be similar in spirit to the approximation of \cite{vu2009differential}. For this approximation, there would exist explicit solutions for $\hat{\theta}$ since one would need to find the roots to a quadratic equation. Once these solutions are found, then one would pick $\hat{\theta}$ such that it lies in $[0,1]$ and minimizes the loss in \eqref{eq:ib}. While this solution is comparable in terms of empirical level and coverage, it has shown less power compared to FIMA (the properties of this approximation are left for future studies).
\end{remk}

Let us now confirm the validity of the FIMA when generating inferential solutions for CIs (and hypothesis tests). To do so, we follow the results in \cite{orso2024accurate} and make use of the tools to determine the properties of extremum estimators \citep{newey1994large} such as the one represented in \eqref{eq:ib}. Based on these properties, we can then determine the validity of the studentized bootstrap (based on the consistency of the estimator $\hat{\pi}$) and, provided that $\hat{\pi} - \pi(\theta)$ converges to a zero-mean symmetric distribution (which is indeed a zero-mean Gaussian in our case), consequently ensures the validity of the percentile bootstrap conditional on $\hat{\pi}$ \citep[see][]{van2000asymptotic, davison1997bootstrap}. In this sense, we establish a required assumption on the true parameter $\theta_0$, on which we aim to perform statistical inference.\\

\begin{assm}[Interior Point]
\label{asm.interior}
The true parameter $\theta_0$ is in the interior of the parameter space $[0,1]$.
\end{assm}

This is a common assumption for matching-based (or extremum) estimators to ensure, among others, differentiability and relative Taylor expansions around the true parameter, as well as the validity of bootstrap solutions. However, to state the validity of FIMA we also define the quantile of the estimator $\hat{\theta}$ as:
\begin{equation}
\label{eq.quantile}
    \hat{\theta}_{\alpha} := \inf \left\{ v \in \mathbb{R} \; : \; \mathbb{P}\left( \hat{\theta} \le v \,\middle|\, \hat{\pi} \right) \ge \alpha \right\}.
\end{equation}
Based on these premises, the following result confirms that the FIMA produces valid CIs.\\

\begin{theorem}
\label{thm.binom}
Letting $\alpha \in (0,1)$, under Assumptions \ref{asm.additive}-\ref{asm.interior} we have
\[
\mathbb{P}\left( \theta_0 \le \hat{\theta}_{\alpha} \right) = \alpha + o_p(1).
\]
\end{theorem}

\begin{proof}
We verify the conditions for the validity of the percentile bootstrap \citep{van2000asymptotic, davison1997bootstrap} and therefore use the results in \cite{newey1994large} for extremum estimators such as the estimator proposed in \eqref{eq:ib}. The first condition is a common regularity condition and requires the parameter space to be convex and compact, in addition to the true parameter being an interior point of this space: indeed the parameter space $[0,1]$ is obviously convex and compact and under Assumption \ref{asm.interior} we trivially respect this condition. The second condition requires $\pi(\theta)$ to be continuously differentiable and injective (one-to-one) on $[0,1]$: we have that $\pi(\theta):=\mathbb{E}[\hat{\pi}(\theta)] = \mathbb{E}[\bar{\theta} + Y] = \mathbb{E}[\bar{\theta}] + \mathbb{E}[Y] = \theta$, and hence $\pi(\theta)$ is injective and continuously differentiable, therefore trivially respecting this condition as well. Defining $\sigma^2(\theta)$ as the asymptotic variance of the DP estimator $\hat{\pi}$, the final condition requires the scaled and centered DP estimator $\sqrt{n}\,\sigma^{-2}(\theta)\,(\hat{\pi} - \pi(\theta))$ to converge, uniformly in $\theta \in [0,1]$, to a zero-mean random variable with a continuous distribution that does not depend on $\theta$: since the DP noise is $O_p(\nicefrac{1}{n})$, we have that 
$$\sqrt{n}\,[\theta_0(1 - \theta_0)]^{-\nicefrac{1}{2}}\,(\hat{\pi} - \theta_0) \overset{D}{\to} \mathcal{N}(0, 1),$$
thereby respecting this condition and thus concluding the proof.
\end{proof}

\begin{remk}
\label{rmk:one_count}
    The FIMA can easily be adapted to the case where, instead of the sample proportion, the privatized quantity is the count itself where, multiplying by $n$ on both sides, Equation \eqref{eq:sample_prop} simply becomes:
    $$\bar{X}_k = \sum_{i = 1}^n \mathbbm{1}(X_i = k),$$
    i.e. the sum of observations in class $k$. In this case, continuing with the binomial (two-class) example and defining $\hat{X} = \bar{X} + Y$ as the DP count for one class (where the sensitivity for the noise is now $\Delta_{\bar{X}} = 1$) it is straightforward to see that:
    $$\check{\pi} = \frac{1}{n}\,\hat{X},$$ is also an unbiased and consistent estimator for $\theta_0$ with asymptotically normal distribution thereby respecting  the conditions for Theorem \ref{thm.binom} to hold. Hence, we have that:
    $$\underset{\theta \in [0,1]}{\argzero} \, [\hat{X} - \hat{X}^*(\theta)] = \underset{\theta \in [0,1]}{\argzero} \,[\check{\pi} - \check{\pi}^*(\theta)],$$
    where $\hat{X}^*(\theta)$ and $\check{\pi}^*(\theta)$ are the simulated counterparts in our matching problem. Therefore we have that the solution to the latter is given by $\hat{\theta} = F_{U^*}^{-1}(\check{\theta}^*)$ where $\check{\theta}^*$ is obtained by developing the matching problem as follows:
    \begin{align*}
        \underbrace{\frac{1}{n}\hat{X}}_{\check{\pi}} &= \underbrace{\frac{1}{n}\hat{X}^*(\theta)}_{\check{\pi}^*(\theta)}\\
        &=\frac{1}{n}\bar{X}(\theta) + \frac{1}{n}Y^* = \bar{\theta}^*(\theta) + \frac{1}{n}Y^*\\
        \implies& \underbrace{\bar{\theta}^*(\theta)}_{F_{U^*}(\theta)} = \underbrace{\frac{1}{n}(\hat{X} - Y^*)}_{\check{\theta}^*},
    \end{align*}
    where $\bar{X}(\theta)$ is the simulated non-private count.
\end{remk}

Let us now consider the more general multinomial case for $K > 2$, i.e. $X\sim f_{\bm{\theta}}^K$ where $\bm{\theta} \in \Delta^{K-1} \subset \mathbb{R}^{K}$. In this case, the non-private estimator of the latter is represented by $\bar{\bm{\theta}} := [\bar{\theta}_1, \hdots, \bar{\theta}_{K}]$ where $\bar{\theta}_k$ is the sample proportion for the $k^{\text{th}}$ class defined in \eqref{eq:sample_prop}. Although this estimator could be privatized in calibrated manners for the specific problem \citep[see e.g.][]{kairouz2014extremal, bhowmick2018protection, wang2019subsampled}, these approaches would not allow to find an explicit solution for the FIMA. Hence, we preserve the component-wise additive mechanism in \eqref{eq.dp_prop} which allows us to define the DP proportion for the $k^{\text{th}}$ class as $\hat{\pi}_k$, with $\hat{\bm{\pi}} : = [\hat{\pi}_1, \hdots, \hat{\pi}_{K}]$. As a result, we can now apply Algorithm \ref{algo_ifb} to each component of the vector $\hat{\bm{\pi}}$ independently to obtain the vector of FIMA solutions $\hat{\bm{\theta}} := [\hat{\theta}_1, \hdots, \hat{\theta}_{K}]$. However, while for the binomial case the FIMA solution lies within the $(0,1)$ interval, as mentioned earlier, this multinomial approach would not guarantee that the vector $\hat{\bm{\pi}}$ lives within the required $(K-1)$-dimensional simplex. As shown further on, this does not necessarily represent a problem for various statistical inferential procedures for categorical data.  Therefore let us first consider the setting where the FIMA solutions $\hat{\bm{\theta}}$ are used to deliver a single test statistic of interest: this can be the $\chi^2$ test statistic (where estimated proportions are used to compute expected counts) or simply selecting a single component of the vector $\hat{\bm\theta}$ to perform inference. In this sense, let us consider a function $\varphi: \mathbb{R}^K \to \mathbb{R}$ that defines such a statistic and make the following assumption on it.\\

\begin{assm}[Differentiable Stat]
\label{asm.diff_stat}
The function $\varphi$ is continuously differentiable with respect to $\bm{\theta}$.
\end{assm}

This assumption is needed since it allows us to use the continuous mapping theorem (and the Delta method) to ensure that $\varphi(\hat{\bm{\theta}})$ converges to $\varphi(\bm{\theta})$, based on the uniform convergence of $\hat{\bm{\theta}}$, as well as to its asymptotic Gaussian distribution. Based on this, we also define $\vartheta:=\varphi(\bm{\theta})$ and, consequently, let $\hat{\vartheta}_{\alpha}$ represent the $\alpha$-level quantile of the distribution of $\hat{\vartheta}:=\varphi(\hat{\bm{\theta}})$ defined similarly to that in \eqref{eq.quantile}. For completeness, we also adapt Assumption \ref{asm.interior} to this multinomial setting.\\

\begin{assm}[Interior Point (Multinomial)]
\label{asm.interior.multinom}
The true parameter $\bm{\theta}_0$ is in the interior of the parameter space $\Theta$.
\end{assm}

Having stated the above assumptions, we can now give the following result.\\

\begin{prop}
\label{prop.multinom}
Letting $\alpha \in (0,1)$, under Assumptions \ref{asm.additive}, \ref{asm.diff_stat} and \ref{asm.interior.multinom} we have
\[
\mathbb{P}\left( \vartheta_0 \le \hat{\vartheta}_{\alpha} \right) = \alpha + o_p(1).
\]
\end{prop}

\begin{proof}
We follow the same steps as the proof of Theorem \ref{thm.binom} and verify the conditions by adapting some definitions. Specifically, since the parameter space $\Theta$ is convex and compact, under Assumption \ref{asm.interior.multinom} we respect the first condition. Similarly to Theorem \ref{thm.binom}, we also see that $\bm{\pi}(\bm{\theta})$ is injective (since $\bm{\pi}(\bm{\theta}) = \bm{\theta}$) and continuously differentiable which, along with Assumption \ref{asm.diff_stat} on the differentiability of $\varphi$, allows us to verify the second condition. Also, since the DP noise is $O_p(\nicefrac{1}{n})$, we have that 
$$\sqrt{n}\,\Sigma(\bm{\theta_0})^{-\nicefrac{1}{2}}(\hat{\bm{\pi}} - \bm{\theta}_0) \overset{D}{\to} \mathcal{N}_K(\bm{0}, \bm{I}),$$ 
where $\Sigma(\bm{\theta}) = \text{diag}(\bm{\theta}) - \bm{\theta}\bm{\theta}^{\top}$, thereby respecting the third condition and concluding the proof.
\end{proof}

Notice that this result implies the validity of the percentile bootstrap for the quantity $\hat{\vartheta}$ which represents a continuously differentiable transform of $\hat{\bm{\theta}}$. The majority of the standard inferential tools for categorical data respect this assumption such as, for example, the $\chi^2$-statistic or logit transforms for the components of $\hat{\bm{\theta}}$ (which allow to estimate the parameters of logistic models with categorical predictors). An exception may be represented by functions (operators) that project the FIMA solution $\hat{\bm{\theta}}$ onto the $(K-1)$-dimensional simplex. More specifically, projections onto the simplex are often not differentiable everywhere and may not be asymptotically Gaussian when close to the boundaries of the simplex. Although this remains problematic in finite samples, we can, however, observe that based on Assumption \ref{asm.interior.multinom} and on the consistency of $\hat{\bm{\theta}}$, the result of Proposition \ref{prop.multinom} still holds for such projections as $n \to \infty$.

\section{Experiments}
\label{sec.experiments}

We investigate the performance of FIMA in different categorical settings, starting from the basic binomial data with one-sample proportion and two-sample proportion tests, and going to the $\chi^2$-test (logistic models are investigated in Appendix \ref{app.experiments}). For all settings, we will use $Z \sim \text{Laplace}(0,1)$, therefore ensuring pure $\epsilon$-DP, and a privacy budget of $\epsilon = 1$. Moreover, for each experimental condition, we run each method $B=10^4$ times to evaluate its statistical properties and, for the FIMA, we set $H=10^3$ which is generally reasonable to have a good approximation for the bootstrap distribution \citep{davison1997bootstrap}. It must be noted that there are few methods that deliver statistical inference guarantees in all these settings, hence we compare FIMA to different approaches across the different settings.

\subsection{One-Sample Proportion}

We first study the one-sample proportion setting for which an optimal DP inferential solution already exists and was put forward in \cite{c9}: we refer to this method as Tulap since it is based on the Truncated-Uniform-Laplace (Tulap) distribution derived in their work. As mentioned in the introduction, this approach delivers the UMP under DP constraints for binomial data \footnote{Code can be found at \url{https://github.com/Zhanyu-Wang/Simulation-based_Finite-sample_Inference_for_Privatized_Data/blob/main/table1/Awan_binomialDP.R}}. Therefore, we compare the FIMA to this method and, as benchmarks, we compare it also to two non-DP options consisting of the standard z-test for proportions (NP) and the exact binomial test (Exact). We focus our attention on the performance of CIs derived from all these methods in small sample settings: we choose to study coverage on samples of size $n = 30$ and for a level $1 - \alpha = 0.95$. We produce $B$ CIs for each method and for different values of $\theta_0$, i.e. $\theta_0 \in \{0.1, 0.11, 0.12, \ldots, 0.985\}$, also to check coverage when moving close to the boundaries of the parameter space. The left plot of Fig. \ref{fig:one_sample_ci} shows the results of this simulation: the non-private benchmarks show the well-known behavior in terms of CI coverage where (i) the z-test suffers from the normal asymptotic approximation for a discrete-natured problem (especially in finite-samples like this simulation) and generally under-covers while (ii) the exact CIs guarantee at least the $1- \alpha$ coverage but also suffers from inconsistent jumps along the range of $\theta_0$ values and tends to over-cover at the boundaries. On the other hand, the two DP CIs both appear to have an exact coverage across all values of $\theta_0$ including the boundaries and are therefore counterintuitively preferable to the non-private benchmarks when it comes to CI coverage. If we focus on the length of the CIs the right plot in Fig. \ref{fig:one_sample_ci} shows that the non-private CIs obviously have shorter intervals across all the range of values $\theta_0$ and therefore are more precise, although only the Exact CIs guarantee coverage. As expected, the two DP approaches have larger CI lengths, with Tulap generally having slightly shorter intervals across the different values of $\theta_0$ compared to FIMA, except towards the boundaries of the parameter space, where the FIMA appears to be marginally more precise. 
\begin{figure}[t]
    \includegraphics[width = \textwidth]{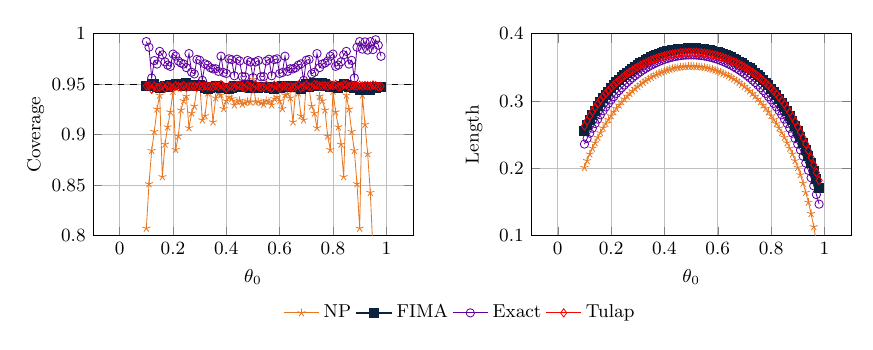}
\caption{Comparison of the coverage (left) and length (right) of $95\%$ CIs for each method and for different parameter values of $\theta_0$.}
\label{fig:one_sample_ci}
\end{figure}
We also investigated the performance of hypothesis tests based on these approaches. For this purpose, we consider the set of hypotheses: $H_0: \theta_0 \geq \gamma$ (null hypothesis) versus $H_A: \theta_0 < \gamma$ (alternative hypothesis), where $0 \leq \gamma \leq 1$ is a constant that we are interested in testing against evidence from the sample. For this experiment, we first study the level (type-I error) of the tests based on these different methods and, for this purpose, set $\gamma = \theta_0$ for different values of $\theta_0 = {0.1, 0.2, \hdots, 0.9}$. Fixing the sample size to $n=30$ to evaluate finite sample performance, in the left plot of Fig. \ref{fig:one_prop_hyp} we can see that all methods are slightly liberal (i.e., reject $H_0$ more than needed) but are all within a reasonable range of the threshold $\alpha = 0.05$ considering the sample size (with Tulap showing the best overall performance). In the right graph, we can observe the power of these methods by fixing $\theta_0 = 0.2$ and varying the value of $\gamma$ from $0.2$ to $1$. It can be observed that the non-private method is obviously the most powerful but is closely followed by the Tulap and FIMA approaches, which have comparable performance in terms of power. In the context of binomial data, our empirical results highlight how the proposed FIMA, while comparing favorably with respect to the standard non-DP approaches, also has a performance comparable to the optimal DP method based on the Tulap distribution of \cite{c9}. Additional results for this setting, with larger and varying sample sizes, can be found in Appendix \ref{app.experiments} where the FIMA shows good performance (often better) than its current alternatives.

\begin{figure}[t]
    \includegraphics[width = \textwidth]{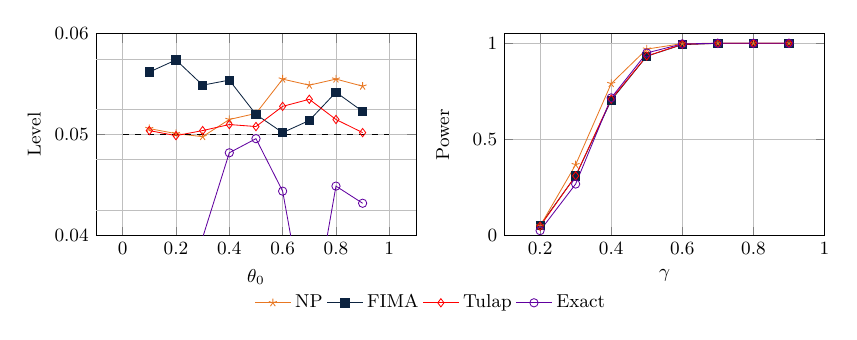}
\caption{Comparison of the level (left) and power (right) of the different test procedures under the one-sided hypothesis test at $\alpha = 0.05$ for different parameter values of $\theta_0$.}
\label{fig:one_prop_hyp}
\end{figure}

\subsection{Two-Sample Proportions}
\label{sec.two_sample}

We now consider the two-sample (independent) proportions test where we are interested in testing $H_0: \theta_1 - \theta_2 \geq 0$ versus $H_A: \theta_1 - \theta_2 < 0$, where $\theta_1$ represents the proportion from the first population and $\theta_2$ that from the second. In this setting, for the FIMA we consider the statistic $\lambda_h = \hat{\theta}_{1,h} - \hat{\theta}_{2,h}$ where $\hat{\theta}_{i,h}$ is the $h^{\text{th}}$ solution of Algorithm \ref{algo_ifb} for the $i^{\text{th}}$ sample. It can be noted that $\lambda_h$ respects Assumption \ref{asm.diff_stat} and therefore fulfills the conditions of Proposition \ref{prop.multinom} (assuming Assumption \ref{asm.interior.multinom} holds). As a DP alternative method, we make use of a recently proposed general hypothesis testing framework under DP put forward by \cite{kazan2023test}, called the Test of Tests (ToT), which compared favorably to other existing alternatives. The latter relies on the sample-and-aggregate approach discussed in \cite{canonne2019structure} and we make use of the optimization procedure proposed in \cite{kazan2023test} to find the best $m$ (number of sub-samples) and $\alpha_0$ (the sub-test significance threshold) \footnote{Code can be found at \url{https://github.com/diff-priv-ht/test-of-tests}}. Of course, we also use the non-private z-test (NP) as a general benchmark for the DP options considered here. Fig. \ref{fig:level_power_two_sample} can be interpreted similarly to the plots in Fig. \ref{fig:one_prop_hyp}. Indeed, the plot on the left shows the level of the considered approaches when the sample size is $n=30$ and where the x-axis represents different values of $\theta_0 = \theta_1 = \theta_2$ (under $H_0$), while the right plot shows the power of the methods when fixing $\theta_1 = 0.2$ and increasing $\theta_2$ from $0.2$ to $0.9$. We can see that the ToT performs very well (indeed it is the best) in terms of type-I error while the NP and FIMA are reasonable when the two parameters are close to $0.5$ but worsen (reject more than needed) when getting closer to the boundaries of the parameter space. The latter is generally expected in small samples but it can be seen that the FIMA is generally worse in these regions, although it still lies within an acceptable level considering the sample size. On the other hand, the ToT performs well across all parameter ranges. When looking at power we see a different picture: the NP test is obviously the most powerful option while the FIMA follows a similar progress below the NP power curve (as expected for a DP approach). For the ToT we consider two versions: one is the actual (empirical) ToT while the other is the theoretical power expected for the ToT \citep[see][]{kazan2023test}. We can observe that the theoretical power has a good performance and is better than the FIMA, while the empirical ToT appears to suffer the small sample setting and is not performing as it should do theoretically.  Putting aside the small sample (which makes subsampling more complicated), it is probably the case that the default optimal parameter approach provided in \cite{kazan2023test} doesn't fit for these categorical settings; hence it may be necessary to also adapt the hyper-parameter tuning procedure for the ToT and therefore these results may not be considered as conclusive for this method for these experiments. However, given also the excellent performance in terms of level, we chose to keep the ToT as a potential alternative to consider in these experiments. Similarly to the one-sample proportion experiments, Appendix \ref{app.experiments} reports further results in this setting with larger and varying sample sizes where the FIMA compares favorably to the considered alternatives.

\begin{figure}[t]
    \includegraphics[width = \textwidth]{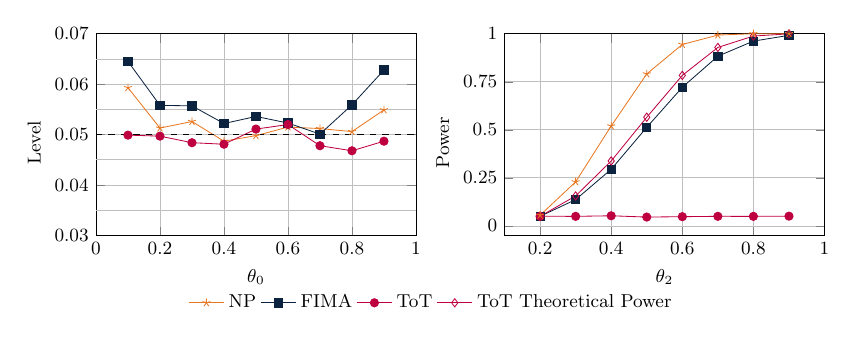}
\caption{Comparison of the level (left) and power (right) of the different test procedures for two-sample hypothesis test at $\alpha = 0.05$ and for different parameter values of $\theta_0$.}
\label{fig:level_power_two_sample}
\end{figure}

\subsection{$\chi^2$-Test}
\label{sec.chi_square}

We perform a final experiment for the setting in which we may be interested in performing a $\chi^2$-test between two categorical variables $X \sim f_{\bm{\theta}_1}^{K_1}$ and $Q \sim f_{\bm{\theta}_2}^{K_2}$ with $K_1$ and $K_2$ classes respectively. In this case, we have to slightly adapt the FIMA to generate a bootstrap distribution of the $\chi^2$ statistics under the null hypothesis of independence between the two categorical variables. Taking into account the distribution of the $n$ total counts across different combinations of the two variables $X$ and $Q$, these can be modeled by a joint multinomial distribution $f_{\bm{\theta}_0}$ with $K = K_1 \cdot K_2$ categories in total and corresponding vector of probabilities $\bm{\theta}_0 = \texttt{vectorize}(\bm{\theta}_1 \bm{\theta}_2^T$), with $\texttt{vectorize}(A)$ representing the vectorization/flattening of the matrix $A$. With this framework in mind, we assume that we obtain privatized counts $\hat{G}_{ij}$ in the form of a contingency table with $i$ and $j$ representing the row and column respectively, with $i = 1, \hdots, K_1$ and $j = 1, \hdots, K_2$. In this case, to perform the $\chi^2$-test we need to slightly adapt the reasoning in Remark \ref{rmk:one_count}. Indeed, we now are computing marginal counts for each variable $X$ and $Q$ which consist in sums of privatized counts $\hat{G}_{ij}$ across the rows and the columns. Hence, taking a marginal count across the rows as an example (i.e. targeting the marginal distribution $f_{\theta_2}^{K_2}$ of $Q$), we have that
$$\hat{Q}_{j} = \sum_{i = 1}^{K_1} \hat{G}_{ij} = \sum_{i = 1}^{K_1} \left(\bar{G}_{ij} + Y_{ij}\right),$$
where $\bar{G}_{ij}$ represents the non-privatized count of the contingency table ($\hat{X}_i$ therefore represents the sum across columns of privatized counts for the $i^{\text{row}}$). As a result, following the same reasoning as Remark 1, we have that the FIMA solution for $\theta_{2j}$ (i.e. the $j^{\text{th}}$ element of $\theta_2$) is given by $\hat{\theta}_{2j} = F_{U^*}^{-1}(\check{\theta}_{2j}^*)$, with $\check{\theta}_{2j}^*$ now being defined as:
$$\check{\theta}_{2j}^* = \frac{1}{n}\left(\hat{Q}_j - \sum_{i = 1}^{K_1} Y_{ij}^*\right),$$
thereby defining the vector of quantities $\check{\bm{\theta}}_2 = [\check{\theta}_{21}, \hdots, \check{\theta}_{2K_2}]$ (with $\check{\bm{\theta}}_1$ denoting the same vector of the other margin). Algorithm \ref{algo_ifb_chi2} provides the overview of this adaptation of the FIMA to this specific test. It must be noted that this is an adaptation specifically for this multinomial context (with sums of privacy noise terms) which would not be required if directly targeting the $K$-dimensional vector $\bm{\theta}_0$ characterizing the individuals counts in the table, in which case the strategy in Remark \ref{rmk:one_count} can directly be applied component-wise. This is in fact the case for logistic regression with categorical predictors in Appendix \ref{app:logistic}.
\begin{algorithm}
\begin{algorithmic}[1]
    \State \textbf{Input:} $\hat{\bm{G}}$: table of DP counts; $n$: data size; $\epsilon$: privacy budget; $H$: number of solutions; $q$: distribution for $Y$; $g$: distribution for $D$; $\delta > 0$: small quantity close to zero.
    \State Compute marginal counts $\hat{\mathbf{X}} = \sum_{j=1}^{K_2} \hat{G}_{ij}$ and $\hat{\mathbf{Q}} = \sum_{i=1}^{K_1} \hat{G}_{ij}$
    \For{$h = 1, \hdots, H$}
    \State Compute $\check{\bm{\theta}}_{1h}^* = \nicefrac{1}{n}(\hat{\mathbf{X}} - \mathbf{Y}_{1h}^*)$ and $\check{\bm{\theta}}_{2h}^* = \nicefrac{1}{n}(\hat{\mathbf{Q}} - \mathbf{Y}_{2h}^*)$, where $\mathbf{Y}_{1,h}^* :=[\sum_{j=1}^{K_2}Y_{1,h,j}^*, \hdots, \sum_{j=1}^{K_2}Y_{K_1, h, j}^*]$ and $\mathbf{Y}_{2,h}^* :=[\sum_{i=1}^{K_1}Y_{1,h,j}^*, \hdots, \sum_{i=1}^{K_1}Y_{K_2, h, j}^*]$, with $Y_{k, h, j} \sim q\left(\frac{2}{n \epsilon}\right)$.
    \State Run lines 4-11 of Algorithm \ref{algo_ifb} on each element of $\check{\bm{\theta}}_{1h}^*$ and $\check{\bm{\theta}}_{2h}^*$ to obtain matrix $\hat{\bm{\theta}}_h = \hat{\bm{\theta}}_{1h}^* \hat{\bm{\theta}}_{2h}^{*T}$
    \State Generate contingency table $\bar{\mathbf{G}}_h^*$ from mulitnomial distribution with parameters $n$ and \texttt{vectorize}($\hat{\bm{\theta}}_h$) 
    \State Apply DP noise to $\bar{\mathbf{G}}_h^*$ using additive noise distribution $q$ to obtain new contingency table of DP counts $\hat{\mathbf{G}}_h^*$
    \State Compute $\chi_h^2$ (the  $\chi^2$-statistic) based on contingency table $\hat{\mathbf{G}}_h^*$ 
    \EndFor
    \State \textbf{Output:} A sequence $\{\chi_h^2\}$ for $h = 1, \hdots, H$
\end{algorithmic}
\caption{FIMA for $\chi^2$-test}
\label{algo_ifb_chi2}
\end{algorithm}

\begin{remk}
    In the case of Algorithm \ref{algo_ifb_chi2}, we must underline that the requirements of Proposition \ref{prop.multinom} only hold when the DP noise comes from a continuous distribution (and not from a discrete distribution). Indeed, in this case the $\chi^2$-statistic is continuously differentiable with respect to $\bm{\theta}_0$ and thus respects Assumption \ref{asm.diff_stat}. Underlining that the results presented in this section are similar also in the discrete case, the theory nevertheless does not currently support this setting.
\end{remk}

\noindent We compare the FIMA in Algorithm \ref{algo_ifb_chi2} to the non-private $\chi^2$-test (NP) as the standard benchmark, and also compare it to DP alternatives which include the ToT used in our previous experiment, as well as the specifically tailored DP $\chi^2$-test put forward by \cite{c5} (we refer to this as Gabo). We consider the same experimental setup as that of \cite{c5} which consists of a simple $2 \times 2$ contingency table generated for different sample sizes ranging from $n = 10$ to $n=20^4$ and with privacy budget $\epsilon = 0.1$ (in \cite{c5} they use $(\epsilon, \delta)-DP$ but we set $\delta = 0$). More specifically, to study the level of the methods, we specify $\bm{\theta}_0 = [0.25, 0.25, 0.25, 0.25]$ under the null $H_0$ (independence) and, for the power under the alternative $H_A$, we set $\bm{\theta}_0 = [0.25, 0.25, 0.25, 0.25] + 0.01 \cdot [1, -1, -1, 1]$. Also, for the FIMA, we use $H = 10^4$ given the higher number of parameters (counts) in this model and, hence, the need to better approximate the FIMA distribution.

\begin{figure}
    \includegraphics[width = \textwidth]{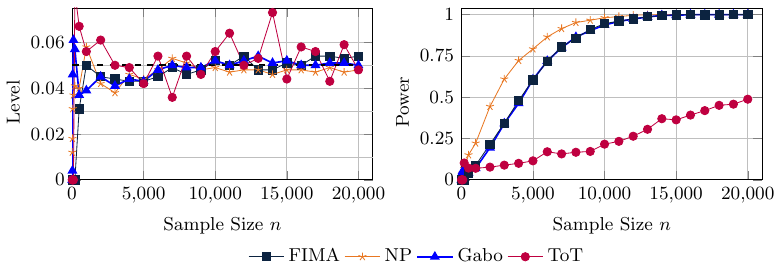}
\caption{Comparison of the level (left) and power (right) of the different test procedures for $\chi^2$-test at $\alpha = 0.05$ for different sample sizes $n$.}
\label{fig:chi2}
\end{figure}

It can be seen in the left plot of Fig. \ref{fig:chi2} how the NP and Gabo fluctuate with more amplitude in smaller sample sizes and then start fluctuating more closely to the desired $\alpha = 0.05$ level as the sample size increases. However, the ToT and FIMA appear to have a lower fluctuation in small samples, while the FIMA appears to be more stable (from the conservative side) around the required level as the sample size increases, the ToT appears to fluctuate with more amplitude. With regard to power, putting aside the NP as the most powerful benchmark, we can see that Gabo and FIMA have very similar performance in terms of power (with Gabo being slightly better overall), whereas the ToT also sees increasing power but less than its alternative, possibly being affected by similar issues as in Sec. \ref{sec.two_sample} thereby probably requiring a better tuning of the hyperparameters $m$ and $\alpha_0$ (among others). 

\subsection{Running Times}

We complete our empirical experiments by highlighting the mean running times of the different approaches with varying sample sizes. Let us therefore consider the one- and two-sample proportion coverage and tests: for the coverage, we record the running time for sample sizes $n~\in~\left\{ 8 \times 2^k : k = 0, \dotsc, 6 \right\} \,\cup\, \{700, 900, 1200\}$ while for tests we have $n~\in~\{16,\ 30,\ 50,\ 100,\ 150,\ 200,\ 350,\ 400,\ 500\}$. Fig. \ref{fig.run_times} shows the log mean running times (in milliseconds) in the different settings for increasing sample sizes. We generally observe that the DP approaches have longer running times, but while the running times for the alternatives generally increase with sample size $n$, we see that the FIMA running time does not appear to increase with sample size, indicating the scalability of the proposed approach. This is mainly due to the choice of $D \sim \text{Beta}(\nicefrac{1}{2}, \nicefrac{1}{2})$ to sample from the intervals, which allows us to directly sample FIMA solutions from the Beta distribution defined in Sec. \ref{sec.fima}.

\begin{figure}[H]
    \includegraphics[width=\textwidth]{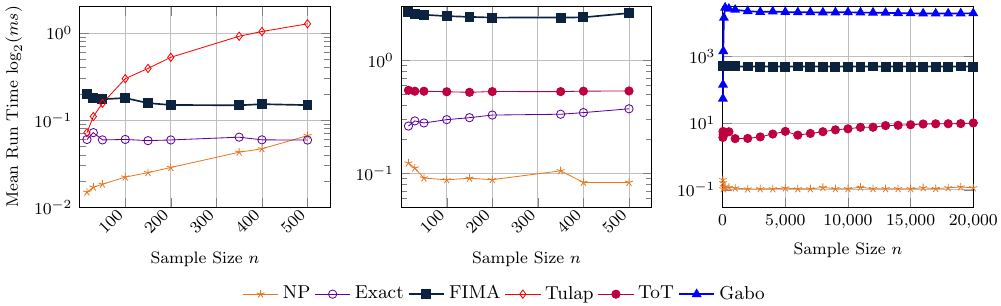}
\caption{Mean log running times (in milliseconds) for all considered methods on different samples sizes for one-proportion hypothesis test (left plot), two-sample hypothesis test (center plot) and $\chi^2$-test (right plot).}
\label{fig.run_times}
\end{figure}  

\section{Conclusion}

The FIMA provides an efficient solution for statistical inference under DP constraints for categorical data characterized by class probabilities $\bm{\theta}$. While providing validity guarantees in terms of its bootstrap distributions, overall it can be observed that it provides a flexible alternative to existing DP approaches in these settings (without the need to specify additional hyper-parameters). Indeed, we cannot claim that the FIMA is the best overall approach based on our experiments, but it generally provides a computationally efficient alternative with stable performance in terms of type-I error and power across a variety of inferential tasks in categorical data settings (Appendix \ref{sec.appendix} also shows some results for logistic models, as well as applications of the FIMA to real-world data from the CDC and ACS). In particular, it can be observed that empirically it has a performance comparable to the UMP test of \cite{c9} while in the other settings it provides a generally safe choice (based solely on our experiments) since it produces good type-I error and is in line with (or better than) the power of some existing DP alternatives. Moreover, based on our theoretical results, different inferential tasks can be performed using FIMA whenever their statistics are a continuously differentiable function of the underlying parameter $\bm{\theta}$, which allows the method to be easily adapted to other categorical data settings.

\section*{Acknowledgements}
In preparing this work, ChatGPT 4o was used for to outline and proofread for some paragraphs. The authors subsequently carefully reviewed and edited the content to ensure accuracy and coherence. They take full responsibility for the integrity of the publication. This work has been funded by NSF-SES 2150615.

\newpage
\bibliographystyle{plainnat}
\bibliography{ref.bib}

\begin{thebibliography}{47}
\providecommand{\natexlab}[1]{#1}
\providecommand{\url}[1]{\texttt{#1}}
\expandafter\ifx\csname urlstyle\endcsname\relax
  \providecommand{\doi}[1]{doi: #1}\else
  \providecommand{\doi}{doi: \begingroup \urlstyle{rm}\Url}\fi

\bibitem[Agresti(2013)]{agresti2013categorical}
Alan Agresti.
\newblock \emph{Categorical data analysis}.
\newblock John Wiley \& Sons, 2013.

\bibitem[Alabi and Vadhan(2022)]{alabi2022hypothesis}
Daniel Alabi and Salil Vadhan.
\newblock Hypothesis testing for differentially private linear regression.
\newblock \emph{Advances in Neural Information Processing Systems},
  35:\penalty0 14196--14209, 2022.

\bibitem[Awan and Slavkovi{\'c}(2018)]{c9}
Jordan Awan and Aleksandra Slavkovi{\'c}.
\newblock Differentially private uniformly most powerful tests for binomial
  data.
\newblock \emph{Advances in Neural Information Processing Systems}, 31, 2018.

\bibitem[Awan and Wang(2023)]{c21}
Jordan Awan and Zhanyu Wang.
\newblock Simulation-based, finite-sample inference for privatized data.
\newblock \emph{arXiv preprint arXiv:2303.05328}, 2023.

\bibitem[Awan and Wang(2024)]{awan2024simulation}
Jordan Awan and Zhanyu Wang.
\newblock Simulation-based, finite-sample inference for privatized data.
\newblock \emph{Journal of the American Statistical Association}, \penalty0
  (just-accepted):\penalty0 1--27, 2024.

\bibitem[Bernstein and Sheldon(2018)]{bernstein2018differentially}
Garrett Bernstein and Daniel~R Sheldon.
\newblock Differentially private {Bayesian} inference for exponential families.
\newblock \emph{Advances in Neural Information Processing Systems},
  31:\penalty0 2919--2929, 2018.

\bibitem[Bernstein and Sheldon(2019)]{bernstein2019differentially}
Garrett Bernstein and Daniel~R Sheldon.
\newblock Differentially private {Bayesian} linear regression.
\newblock \emph{Advances in Neural Information Processing Systems},
  32:\penalty0 525--535, 2019.

\bibitem[Bhowmick et~al.(2018)Bhowmick, Duchi, Freudiger, Kapoor, and
  Rogers]{bhowmick2018protection}
Abhishek Bhowmick, John Duchi, Julien Freudiger, Gaurav Kapoor, and Ryan
  Rogers.
\newblock Protection against reconstruction and its applications in private
  federated learning.
\newblock \emph{arXiv preprint arXiv:1812.00984}, 2018.

\bibitem[Canonne et~al.(2019)Canonne, Kamath, McMillan, Smith, and
  Ullman]{canonne2019structure}
Cl{\'e}ment~L Canonne, Gautam Kamath, Audra McMillan, Adam Smith, and Jonathan
  Ullman.
\newblock The structure of optimal private tests for simple hypotheses.
\newblock In \emph{Proceedings of the 51st Annual ACM SIGACT Symposium on
  Theory of Computing}, pages 310--321, 2019.

\bibitem[Davison and Hinkley(1997)]{davison1997bootstrap}
Anthony~Christopher Davison and David~Victor Hinkley.
\newblock \emph{Bootstrap methods and their application}.
\newblock Number~1. Cambridge university press, 1997.

\bibitem[Dwork(2006{\natexlab{a}})]{c20}
Cynthia Dwork.
\newblock Differential privacy.
\newblock In \emph{International colloquium on automata, languages, and
  programming}, pages 1--12. Springer, 2006{\natexlab{a}}.

\bibitem[Dwork(2006{\natexlab{b}})]{dwork2006differential}
Cynthia Dwork.
\newblock Differential privacy.
\newblock In \emph{International colloquium on automata, languages, and
  programming}, pages 1--12. Springer, 2006{\natexlab{b}}.

\bibitem[Dwork et~al.(2006)Dwork, McSherry, Nissim, and
  Smith]{dwork2006calibrating}
Cynthia Dwork, Frank McSherry, Kobbi Nissim, and Adam Smith.
\newblock Calibrating noise to sensitivity in private data analysis.
\newblock In \emph{Theory of cryptography conference}, pages 265--284.
  Springer, 2006.

\bibitem[Efron(2012)]{efron2012bayesian}
Bradley Efron.
\newblock Bayesian inference and the parametric bootstrap.
\newblock \emph{The annals of applied statistics}, 6\penalty0 (4):\penalty0
  1971, 2012.

\bibitem[Efron and Tibshirani(1994)]{c23}
Bradley Efron and Robert~J Tibshirani.
\newblock \emph{An introduction to the bootstrap}.
\newblock Chapman and Hall/CRC, 1994.

\bibitem[Ferrando et~al.(2022)Ferrando, Wang, and
  Sheldon]{ferrando2022parametric}
Cecilia Ferrando, Shufan Wang, and Daniel Sheldon.
\newblock Parametric bootstrap for differentially private confidence intervals.
\newblock In \emph{International Conference on Artificial Intelligence and
  Statistics}, pages 1598--1618. PMLR, 2022.

\bibitem[Fishman(2013)]{c24}
George Fishman.
\newblock \emph{Monte Carlo: concepts, algorithms, and applications}.
\newblock Springer Science \& Business Media, 2013.

\bibitem[Gaboardi et~al.(2016)Gaboardi, Lim, Rogers, and Vadhan]{c5}
Marco Gaboardi, Hyun Lim, Ryan Rogers, and Salil Vadhan.
\newblock Differentially private chi-squared hypothesis testing: Goodness of
  fit and independence testing.
\newblock In \emph{International conference on machine learning}, pages
  2111--2120. PMLR, 2016.

\bibitem[Geng and Viswanath(2015)]{geng2015optimal}
Quan Geng and Pramod Viswanath.
\newblock The optimal noise-adding mechanism in differential privacy.
\newblock \emph{IEEE Transactions on Information Theory}, 62\penalty0
  (2):\penalty0 925--951, 2015.

\bibitem[Gong(2019)]{gong2019exact}
Ruobin Gong.
\newblock Exact inference with approximate computation for differentially
  private data via perturbations.
\newblock \emph{arXiv preprint arXiv:1909.12237}, 2019.

\bibitem[Gourieroux et~al.(1993)Gourieroux, Monfort, and Renault]{c13}
Christian Gourieroux, Alain Monfort, and Eric Renault.
\newblock Indirect inference.
\newblock \emph{Journal of applied econometrics}, 8\penalty0 (S1):\penalty0
  S85--S118, 1993.

\bibitem[Guerrier et~al.(2019)Guerrier, Dupuis-Lozeron, Ma, and
  Victoria-Feser]{c14}
St{\'e}phane Guerrier, Elise Dupuis-Lozeron, Yanyuan Ma, and Maria-Pia
  Victoria-Feser.
\newblock Simulation-based bias correction methods for complex models.
\newblock \emph{Journal of the American Statistical Association}, 114\penalty0
  (525):\penalty0 146--157, 2019.

\bibitem[Guerrier et~al.(2020)Guerrier, Karemera, Orso, Victoria-Feser, and
  Zhang]{c15}
St{\'e}phane Guerrier, Mucyo Karemera, Samuel Orso, Maria-Pia Victoria-Feser,
  and Yuming Zhang.
\newblock A general approach for simulation-based bias correction in high
  dimensional settings.
\newblock \emph{arXiv preprint arXiv:2010.13687}, 2020.

\bibitem[Hannig(2009)]{c27}
Jan Hannig.
\newblock On generalized fiducial inference.
\newblock \emph{Statistica Sinica}, pages 491--544, 2009.

\bibitem[Hannig et~al.(2016)Hannig, Iyer, Lai, and Lee]{c25}
Jan Hannig, Hari Iyer, Randy~CS Lai, and Thomas~CM Lee.
\newblock Generalized fiducial inference: A review and new results.
\newblock \emph{Journal of the American Statistical Association}, 111\penalty0
  (515):\penalty0 1346--1361, 2016.

\bibitem[Ju et~al.(2022)Ju, Awan, Gong, and Rao]{ju2022data}
Nianqiao Ju, Jordan Awan, Ruobin Gong, and Vinayak Rao.
\newblock Data augmentation {MCMC} for {Bayesian} inference from privatized
  data.
\newblock \emph{Advances in Neural Information Processing Systems},
  35:\penalty0 12732--12743, 2022.

\bibitem[Kairouz et~al.(2014)Kairouz, Oh, and Viswanath]{kairouz2014extremal}
Peter Kairouz, Sewoong Oh, and Pramod Viswanath.
\newblock Extremal mechanisms for local differential privacy.
\newblock \emph{Advances in neural information processing systems}, 27, 2014.

\bibitem[Karwa et~al.(2015)Karwa, Kifer, and Slavkovi{\'c}]{karwa2015private}
Vishesh Karwa, Dan Kifer, and Aleksandra~B Slavkovi{\'c}.
\newblock Private posterior distributions from variational approximations.
\newblock \emph{arXiv preprint arXiv:1511.07896}, 2015.

\bibitem[Kazan et~al.(2023{\natexlab{a}})Kazan, Shi, Groce, and Bray]{c11}
Zeki Kazan, Kaiyan Shi, Adam Groce, and Andrew~P Bray.
\newblock The test of tests: A framework for differentially private hypothesis
  testing.
\newblock In \emph{International Conference on Machine Learning}, pages
  16131--16151. PMLR, 2023{\natexlab{a}}.

\bibitem[Kazan et~al.(2023{\natexlab{b}})Kazan, Shi, Groce, and
  Bray]{kazan2023test}
Zeki Kazan, Kaiyan Shi, Adam Groce, and Andrew~P Bray.
\newblock The test of tests: A framework for differentially private hypothesis
  testing.
\newblock In \emph{International Conference on Machine Learning}, pages
  16131--16151. PMLR, 2023{\natexlab{b}}.

\bibitem[Kifer and Rogers(2016)]{c3}
Daniel Kifer and Ryan Rogers.
\newblock A new class of private chi-square tests.
\newblock In \emph{Proceedings of the 20th International Conference on
  Artificial Intelligence and Statistics, AISTATS}, volume~17, pages 991--1000,
  2016.

\bibitem[Newey and McFadden(1994)]{newey1994large}
Whitney~K Newey and Daniel McFadden.
\newblock Large sample estimation and hypothesis testing.
\newblock \emph{Handbook of econometrics}, 4:\penalty0 2111--2245, 1994.

\bibitem[Nissim et~al.(2007)Nissim, Raskhodnikova, and Smith]{c10}
Kobbi Nissim, Sofya Raskhodnikova, and Adam Smith.
\newblock Smooth sensitivity and sampling in private data analysis.
\newblock In \emph{Proceedings of the thirty-ninth annual ACM symposium on
  Theory of computing}, pages 75--84, 2007.

\bibitem[Orso et~al.(2024)Orso, Karemera, Victoria-Feser, and
  Guerrier]{orso2024accurate}
Samuel Orso, Mucyo Karemera, Maria-Pia Victoria-Feser, and St{\'e}phane
  Guerrier.
\newblock An accurate percentile method for parametric inference based on
  asymptotically biased estimators.
\newblock \emph{arXiv preprint arXiv:2405.05403}, 2024.

\bibitem[Pe{\~n}a and Barrientos(2022)]{c12}
V{\'\i}ctor Pe{\~n}a and Andr{\'e}s~F Barrientos.
\newblock Differentially private hypothesis testing with the subsampled and
  aggregated randomized response mechanism.
\newblock \emph{arXiv preprint arXiv:2208.06803}, 2022.

\bibitem[Robert(2016)]{c26}
Christian~P Robert.
\newblock Approximate bayesian computation: a survey on recent results.
\newblock In \emph{Monte Carlo and Quasi-Monte Carlo Methods: MCQMC, Leuven,
  Belgium, April 2014}, pages 185--205. Springer, 2016.

\bibitem[Uhlerop et~al.(2013)Uhlerop, Slavkovi{\'c}, and Fienberg]{c8}
Caroline Uhlerop, Aleksandra Slavkovi{\'c}, and Stephen~E Fienberg.
\newblock Privacy-preserving data sharing for genome-wide association studies.
\newblock \emph{The Journal of privacy and confidentiality}, 5\penalty0
  (1):\penalty0 137, 2013.

\bibitem[Van~der Vaart(2000)]{van2000asymptotic}
Aad~W Van~der Vaart.
\newblock \emph{Asymptotic statistics}, volume~3.
\newblock Cambridge university press, 2000.

\bibitem[Vu and Slavkovic(2009{\natexlab{a}})]{c6}
Duy Vu and Aleksandra Slavkovic.
\newblock Differential privacy for clinical trial data: Preliminary
  evaluations.
\newblock In \emph{2009 IEEE International Conference on Data Mining
  Workshops}, pages 138--143. IEEE, 2009{\natexlab{a}}.

\bibitem[Vu and Slavkovic(2009{\natexlab{b}})]{vu2009differential}
Duy Vu and Aleksandra Slavkovic.
\newblock Differential privacy for clinical trial data: Preliminary
  evaluations.
\newblock In \emph{2009 IEEE International Conference on Data Mining
  Workshops}, pages 138--143. IEEE, 2009{\natexlab{b}}.

\bibitem[Wang et~al.(2019)Wang, Balle, and Kasiviswanathan]{wang2019subsampled}
Yu-Xiang Wang, Borja Balle, and Shiva~Prasad Kasiviswanathan.
\newblock Subsampled r{\'e}nyi differential privacy and analytical moments
  accountant.
\newblock In \emph{The 22nd international conference on artificial intelligence
  and statistics}, pages 1226--1235. PMLR, 2019.

\bibitem[Wang et~al.(2015)Wang, Lee, and Kifer]{C7}
Yue Wang, Jaewoo Lee, and Daniel Kifer.
\newblock Revisiting differentially private hypothesis tests for categorical
  data.
\newblock \emph{arXiv preprint arXiv:1511.03376}, 2015.

\bibitem[Wang et~al.(2018{\natexlab{a}})Wang, Kifer, Lee, and Karwa]{c22}
Yue Wang, Daniel Kifer, Jaewoo Lee, and Vishesh Karwa.
\newblock Statistical approximating distributions under differential privacy.
\newblock \emph{Journal of Privacy and Confidentiality}, 8\penalty0 (1),
  2018{\natexlab{a}}.

\bibitem[Wang et~al.(2018{\natexlab{b}})Wang, Kifer, Lee, and
  Karwa]{wang2018statistical}
Yue Wang, Daniel Kifer, Jaewoo Lee, and Vishesh Karwa.
\newblock Statistical approximating distributions under differential privacy.
\newblock \emph{Journal of Privacy and Confidentiality}, 8\penalty0 (1),
  2018{\natexlab{b}}.

\bibitem[Williams and McSherry(2010)]{williams2010probabilistic}
Oliver Williams and Frank McSherry.
\newblock Probabilistic inference and differential privacy.
\newblock \emph{Advances in Neural Information Processing Systems},
  23:\penalty0 2451--2459, 2010.

\bibitem[Xie and Wang(2022)]{xie2022repro}
Min-ge Xie and Peng Wang.
\newblock Repro samples method for finite-and large-sample inferences.
\newblock \emph{arXiv preprint arXiv:2206.06421}, 2022.

\bibitem[Zhang et~al.(2022)Zhang, Ma, Orso, Karemera, Victoria-Feser, and
  Guerrier]{c16}
Yuming Zhang, Yanyuan Ma, Samuel Orso, Mucyo Karemera, Maria-Pia
  Victoria-Feser, and St{\'e}phane Guerrier.
\newblock Just identified indirect inference estimator: Accurate inference
  through bias correction.
\newblock \emph{arXiv preprint arXiv:2204.07907}, 2022.

\end{thebibliography}

\newpage

\section{Appendix}
\label{sec.appendix}

\subsection{Additional Experimental Results}
\label{app.experiments}

We explore a few more experimental settings in addition to those in Sec. \ref{sec.experiments}, in particular we investigate the behavior of the different methods with larger and varying sample sizes. We first study the empirical performance of the approaches for the one-sample proportion case to understand how CI coverage (at level 0.95) and lengths behave with increasing sample sizes. We look at sample sizes in the following range: $n = \left\{ 8 \times 2^k : k = 0, \dotsc, 6 \right\} \,\cup\, \{700, 900, 1200\}$. The results of our experiments are shown in Fig. \ref{fig.one_sample_vary_n} where the left plot shows coverage and the right one shows lengths.

\begin{figure}[H]
    \includegraphics{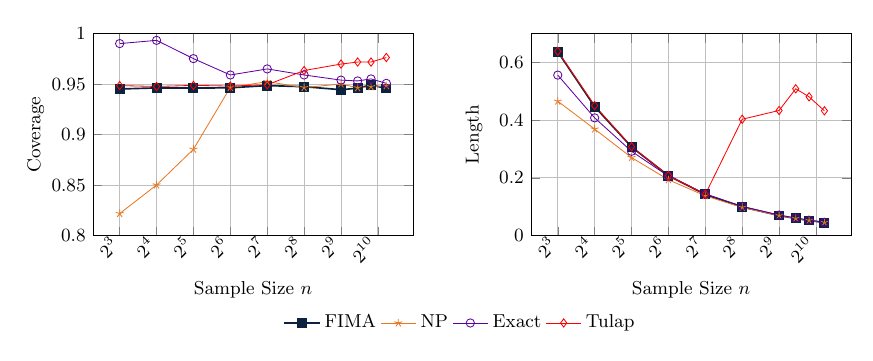}
\caption{One-sample CI coverage (left) and length (right) with different sample sizes}
\label{fig.one_sample_vary_n}
\end{figure}

We can see how the FIMA shows a consistent and good coverage across all sample sizes, while the non-private CIs converge to the nominal level (from above and below respectively) as the sample size increases. We can see however that, as the sample size increases, at a certain point the Tulap starts to over-cover. The latter corresponds to the point in the right plot where the Tulap CI lengths start to increase as well. This behavior of the Tulap, as for the ToT in other experiments, could be due to numerical issues in their implementation and should therefore not be considered as conclusive on the validity of these alternatives. This being said, in these experimental settings the FIMA CI lengths instead do indeed converge to the lengths of the non-private CIs as the sample size increases while maintaining good coverage.

We also investigate the behaviour of the approaches considered for the two-sample proportion setting with a larger sample size $n_1 = n_2 = 100$ (instead of $30$). The results for level and power of the same test as in Sec. \ref{sec.two_sample} are presented in Fig.  \ref{fig:two_prop_hyp_n100}.

\begin{figure}[H]
    \includegraphics{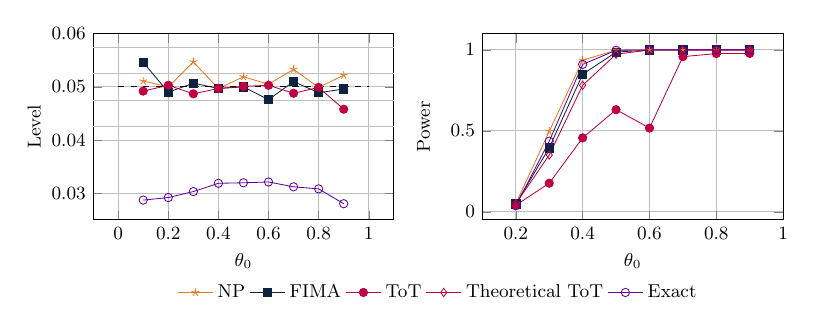}
\caption{Level (left plot) and power (right plot) evaluation of different test procedures under a one-sided hypothesis test for two proportions with  $n = 100$.}
\label{fig:two_prop_hyp_n100}
\end{figure}

It can be observed how all tests improve their level (left plot), especially FIMA and NP (non-private) while the ToT remains as good as in the smaller sample case. However, the right plot shows that the power of the FIMA improves and gets closer to the NP power compared to the theoretical ToT. The empirical ToT can be seen to greatly improve with respect to the small sample setting of $n = 30$ (see Sec. \ref{sec.two_sample}), although it still has lower power than what it should achieve theoretically (it would therefore probably require a better calibration of its hyper-parameters for these problems).

We finally look at the behavior of hypothesis testing with increasing sample sizes (compared to the results when varying parameter values). Specifically we study the performance of the different approaches considered for one- and two-sample proportion tests. In the one-sample case, we fix $\theta_0 = 0.95$ and test the hypotheses $H_0: \theta_0 \leq 0.9$ versus $H_A: \theta_0 > 0.9$, while for the two-sample setting we fix $\theta_1 = 0.8$ and $\theta_2 = 0.9$ and test the hypotheses $H_0: \theta_1 \geq \theta_2$ versus $H_A: \theta_1 < \theta_2$.

\begin{figure}[H]
    \includegraphics{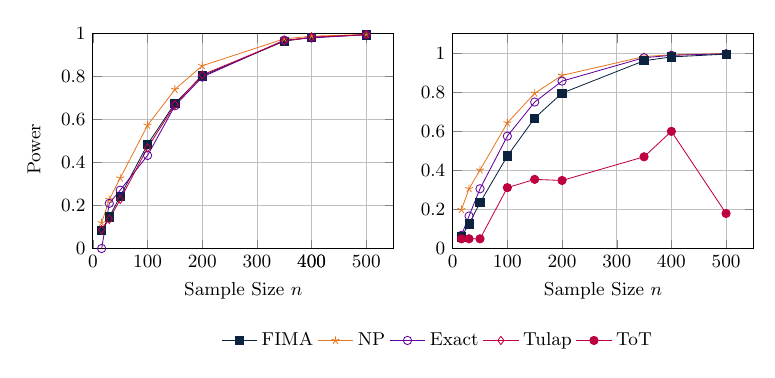}
\caption{Power with increasing sample sizes for one-sample proportion test (left plot) and two-sample proportion test (right plot) of different test procedures under a one-sided hypothesis test.}
\label{fig:two_prop_hyp_n100}
\end{figure}

\subsection{Logistic Models with Categorical Predictors}
\label{app:logistic}

We run experiments when using FIMA for inference on the parameters of logistic regression models with categorical predictors. In particular, following \cite{agresti2013categorical} we know that for a multinomial model with $K$ classes and choosing the last as the reference class, then we have the following relationships:
\begin{align*}
    &\text{class 1} \qquad  \qquad\text{logit}(\theta_1) = \beta_0 + \beta_1\\
    &\dots \qquad \dots \qquad \dots\\
    &\text{class} ~ K-1 \qquad \text{logit}(\theta_{K-1}) = \beta_0 + \beta_{K-1}\\
     &\text{class} ~ K \qquad  \qquad \text{logit}(\theta_{K}) = \beta_0,
\end{align*}
where $\beta_k$ is the logistic model coefficient for the $k^{\text{th}}$ class. As a consequence we have that:
\begin{equation}
\label{eq.logit_beta}
    \beta_0 = \text{logit} (\theta_K) \quad \mbox{and} \quad \beta_k = \text{logit} (\theta_k) - \text{logit} (\theta_K),
\end{equation}
which allows us to make use of a plug-in approach replacing the FIMA solution $\hat{\theta}_k$ in the above equations to obtain FIMA solutions $\hat{\beta}_k$. Indeed, the logit function is continuously differentiable therefore respecting Assumption \ref{asm.diff_stat}. To the best of our knowledge, there are no currently existing DP methods that can directly be applied to this setting and therefore limit ourselves to showing the results for FIMA. As for the $\chi^2$-test, considering the larger number of parameters (counts) in this setting, we set $H = 10^4$ increasing it to better approximate the FIMA distribution in these settings.

\subsubsection{One Binary Predictor}

Consider the simple logistic model regression model for binary response $X$ and with one binary predictor $T$ such that the model can be written as
$$\text{logit}(\theta_{k_T}) = \beta_0 + \beta_1\,T,$$
where $\theta_{k_1} := \mathbb{P}(X=1 \,|\, T = 1)$ and $\theta_{k_0} := \mathbb{P}(X=1 \,|\, T = 0)$. Therefore, we apply Algorithm \ref{algo_ifb} to obtain FIMA solutions $\hat{\theta}_{k_1}$ and $\hat{\theta}_{k_0}$ respectively and plug these into the formula in \eqref{eq.logit_beta}. We keep the same parameters of our previous experiments and fix $\alpha=0.05$, $\epsilon = 1$ and $B = 10^4$. Table \ref{tab:logit_one_pred} shows the CI coverage and length for the parameters $\beta_0$ and $\beta_1$.

\begin{table}[H]

\begin{tabular}{c cc cc}
\toprule
$n$ & \multicolumn{2}{c}{$\beta_0$} & \multicolumn{2}{c}{$\beta_1$} \\
     & Coverage & Length & Coverage & Length \\
\midrule
15   & 0.9518 & 5.6602 & 0.9761 & 9.4099 \\
30   & 0.9493 & 3.9519 & 0.9724 & 6.9174 \\
100  & 0.9503 & 1.1586 & 0.9693 & 3.8286 \\
200  & 0.9495 & 0.8799 & 0.9464 & 2.0628 \\
500  & 0.9499 & 0.5377 & 0.9483 & 1.1215 \\
1000 & 0.9547 & 0.3696 & 0.9496 & 0.7670 \\
2000 & 0.9491 & 0.2607 & 0.9453 & 0.5299 \\
\bottomrule
\end{tabular}
\vspace{1em}
\centering
\caption{CI coverage and lengths of the FIMA for logistic regression with one binary predictor for different sample sizes.}
\label{tab:logit_one_pred}
\end{table}

We can see that for all sample sizes the coverage is close to the nominal level of 0.95, or at least approaches it from the conservative side in the case of $\beta_1$.

\subsection{Two Binary Predictors}

We apply the same FIMA mechanism for the case of two binary predictors, $T_1$ and $T_2$ such that the model is defined as:
$$\text{logit}(\theta_{k_\mathbf{T}}) = \beta_0 + \beta_1\,T_1 + \beta_2 \, T_2,$$
with $\mathbf{T}:= [T_1 \, T_2]$. The results for this setting are presented in Table \ref{tab:logit_two_pred}.

\begin{table}[H]
\centering
\begin{tabular}{c cc cc cc}
\toprule
$n$ & \multicolumn{2}{c}{$\beta_0$} & \multicolumn{2}{c}{$\beta_1$} & \multicolumn{2}{c}{$\beta_2$} \\
    & Coverage & Length & Coverage & Length & Coverage & Length \\
\midrule
15   & 0.9867 & 9.52  & 0.9858 & 13.94 & 0.9885 & 13.36 \\
30   & 0.9504 & 4.68  & 0.9654 & 10.05 & 0.9766 & 9.94  \\
100  & 0.9491 & 1.51  & 0.9747 & 4.93  & 0.9629 & 4.50  \\
200  & 0.9495 & 1.24  & 0.9592 & 3.42  & 0.9481 & 2.22  \\
500  & 0.9525 & 0.77  & 0.9472 & 1.70  & 0.9540 & 1.31  \\
1000 & 0.9513 & 0.53  & 0.9452 & 1.19  & 0.9601 & 0.86  \\
2000 & 0.9512 & 0.37  & 0.9413 & 0.77  & 0.9583 & 0.62  \\
\bottomrule
\end{tabular}
\vspace{1em}
\centering
\caption{CI coverage and lengths of the FIMA for logistic regression with two binary predictors for different sample sizes.}
\label{tab:logit_two_pred}
\end{table}

Also in this case, coverage and lengths for the $\beta_i$ parameters appear to be good with CI lengths decreasing with sample size.

\subsection{Applications}

We apply the FIMA to real-world datasets, validating its practical utility. First, we examine the HIV infection data from the United States Centers for Disease Control and Prevention (CDC) for 2022\footnote{\url{https://www.cdc.gov/hiv/data-research/facts-stats/index.html}}. Publicly reporting these data is crucial for guiding interventions, yet individuals represented in summary statistics may be re-identified, particularly in small or rural populations. We apply the two-sample test of proportions with FIMA to compare HIV infection rates between sub-populations in this dataset, while ensuring individual privacy.  In each application, we compare the results of FIMA with those obtained from a standard non-private (public) approach.


\subsubsection{HIV Diagnoses in the US}
In 2022, the CDC reported 37,981 new HIV diagnoses among individuals aged 13 and above in the U.S., with 70\% (26,749) among gay, bisexual, and other men reporting male-to-male sexual contact. Among these, 9,374 are Hispanic/Latino and 8,831 are Black/African American. For targeted interventions, it is essential to assess which of these groups is more affected by HIV. We conduct a private test to determine if the proportion of HIV-positive Hispanic/Latino individuals reporting male-to-male sexual contact ($\theta_1$) significantly exceeds that of their Black/African American counterparts ($\theta_2$), leading to the set of hypotheses $H_0: \theta_1 \leq \theta_2$ and $H_A: \theta_1 > \theta_2$.

We applied the Laplace mechanism to privatize observed proportions with a budget allocation of \(\nicefrac{\epsilon}{2}\). These privatized proportions, the sample size (\(n = 37,981\)), \(\epsilon\), and \(H\) were inputs for our FIMA two-sample test. With \(\epsilon = 0.1\) and \(H = 10^3\), we obtained \(p\)-value \(= 0\). Results from the non-private asymptotic test (\(p\)-value \(= 9.5 \times 10^{-7}\)) and Fisher's exact binomial test (\(p\)-value \(= 9.9 \times 10^{-7}\)) aligned with our findings. To evaluate consistency, we repeated the private test $10^4$ times with different random seeds, observing a minimum \(p\)-value of \(0\) and a maximum of \(0.0062\). We further analyzed the effect of \(\epsilon\) on test outcomes. Table \ref{ep_pvalue} shows the \(p\)-values for different privacy budgets, \(\epsilon\), and \(H = 10^3\), demonstrating that conclusions align with non-private tests even for small \(\epsilon\) values, down to \(\epsilon = 0.01\).

\begin{table}[h!]
\centering
\resizebox{0.9\textwidth}{!}{%
\begin{tabular}{ccccccccc}
\toprule
\(\epsilon\) & 0.001 & 0.01 & 0.1 & 0.5 & 1 & 3 & 5 & 10 \\
\midrule

p-value & 0.7634 & 0.0006 & 0.0000 & 0.0000 & 0.0000 & 0.0000 & 0.0000 & 0.0000\\
\bottomrule
\end{tabular}%
}
\vspace{1em}
\caption{Values of \(\epsilon\) and corresponding FIMA p-values of two-sample proportion tests on HIV infections data for Hispanic/Latino origin and Black/African American individuals who reported male-to-male sexual contact.}
\label{ep_pvalue}
\end{table}

\subsubsection{HIV Diagnoses}
The CDC aims to reduce new HIV infections to 9,300 by 2025 and 3,000 by 2030. Effective interventions are required at multiple levels, especially for highly affected groups, such as those reporting male-to-male sexual contact. In Alabama, the Office of HIV Prevention and Care reported 374 new HIV cases in the third quarter of 2024, with 232 involving male-to-male sexual contact\footnote{\url{https://www.alabamapublichealth.gov/hiv/statistics.html}}. To evaluate progress toward the CDC’s goals, we tested if the HIV infection rate in this subpopulation in Alabama is significantly lower than the 2022 national rate (0.7).

Using the FIMA one-sample test of proportions we tested \( H_0: \theta_0 \geq 0.7 \) versus \( H_1: \theta_0 < 0.7 \) with a privacy budget of \( \epsilon = 1 \) and \( H = 10^4 \). This yielded a \( p \)-value of 0.0002, compared to 0.0004 from the non-private test. A 95\% differentially private confidence interval for the HIV infection rate among those reporting male-to-male sexual contact in Alabama is (0.5652, 0.6637), closely aligning with the non-private interval of (0.5711, 0.6695). Our results therefore indicate substantial progress toward the CDC’s target for reducing HIV diagnoses, with the private and non-private tests producing consistent conclusions.

\subsubsection{Poverty Status and Race in the US}

We use 2022 American Community Survey (ACS) data\footnote{See URL in Appendix.} to test for dependence between race and poverty status in the US. Our FIMA $\chi^2$-test of independence (see Sec. \ref{sec.chi_square}) ensures that this analysis can be conducted without compromising individual privacy. Testing this hypothesis on national data with a privacy budget of \( \epsilon = 0.00001 \) and \( H = 10^3 \) yielded a \( p \)-value of \( 0 \), consistent with the non-private test. Table \ref{chisquare_table} displays \( p \)-values for the FIMA $\chi^2$-test across different privacy budgets \( \epsilon \). With \( \epsilon = 0.00001 \), we repeated the private test $10^4$ times, obtaining a minimum \( p \)-value of \( 0 \) and a maximum of \( 0.004 \).

\begin{table}[h!]
\centering
\resizebox{\textwidth}{!}{%
\begin{tabular}{cccccccccccc}
\toprule
\(\epsilon\) & 0.000001 & 0.00001 & 0.0001 & 0.001 & 0.01 & 0.1 & 0.5 & 1 & 3 & 5 & 10 \\
\midrule
p-value & 0.532 & 0.000 & 0.000 & 0.000 & 0.000 & 0.000 & 0.000 & 0.000 & 0.000 & 0.000 & 0.000 \\
\bottomrule
\end{tabular}%
}
\vspace{1em}
\caption{$\chi^2$-test \( p \)-value for different privacy budget values (\(\epsilon\)).}
\label{chisquare_table}
\end{table}

\vspace{0.1cm}
\noindent We extended this analysis to state and zipcode levels, using data from Alabama and the zipcode 36830. For Alabama, the non-private test returned a \( p \)-value of \( 0 \); however, the zipcode-level data had cells with zero counts, preventing the non-private test from running. With \( \epsilon = 0.001 \) and \( H = 1000 \), our private test yielded \( p \)-values of \( 0 \) for both state and zipcode levels. Thus, the FIMA $\chi^2$-test aligns with the non-private test where applicable and remains robust even when the non-private test is infeasible.

\end{document}